\documentclass[a4paper,12pt]{article}
\usepackage{amsmath,amssymb,amsthm,graphics,latexsym,amsfonts}
\usepackage{fancyhdr}
\usepackage{color}
\usepackage{graphics}
\usepackage{epstopdf}
\usepackage{amssymb}
\usepackage{cite}
\usepackage{hyperref}
\hypersetup{
    colorlinks=true,
    linkcolor=cyan,
    filecolor=cyan,
    urlcolor=cyan,
    citecolor=cyan,
}

\title{\Large Maximum Diminished Sombor Index of Molecular Trees with a Perfect Matching\,\thanks{The work was supported by the Doctoral Research Foundation of Xinjiang Normal University(No.XJNUZBS2446), the Natural Science Foundation of Xinjiang Uygur Autonomous Region under [Grant No. 2024D01B54] and the Tianchi Talent Training Program of Xinjiang Uygur Autonomous Region.}}
\author{ {Fei Guo$^a$, Fangxia Wang$^b$}\\
\small  $^a$College of Mathematics and System Sciences, Xinjiang
University \\ \small  Urumqi, Xinjiang 830046, P.R.China \\
\small $^b$ School of Mathematical Science, Xinjiang Normal University\\
\small Urumqi 830054, P.R.China\\}
\date{}

\newtheorem{theorem}{Theorem}[section]
\newtheorem{lemma}[theorem]{Lemma}

\usepackage{indentfirst}

\begin{document}

\maketitle {\small \noindent{\bfseries Abstract}
The diminished Sombor index $(DSO)$ of a graph $G$, introduced by Rajathagiri, is defined as $$DSO(G)=\sum_{uv\in E}\frac{\sqrt{d_u^2+d_v^2}}{d_u+d_v},$$
where $d_u$ and $d_v$ are the degrees of vertices $u$ and $v$. A graph $G$ is a molecular graph if $d_G(u)\leq 4$ for all $u\in V(G)$. In this paper, we examine the chemical applicability of the $DSO$ index for predicting physicochemical properties of octane isomers. We also determine the maximum value of the diminished Sombor index among all molecular trees of order $n$ with perfect matching and characterize all the corresponding extremal trees. 

\section {\large Introduction}

A single number, representing a chemical structure in graph-theoretical terms via the molecular graph, is called a topological index if it correlates with a molecular property. Topological indices are used to understand physicochemical properties of chemical compounds. By a molecular graph we understand a simple graph, where the vertices represent the carbon atoms and its edges represent the carbon-carbon bonds.
The topological index, also known as a molecular descriptor, is crucial in determining this link with no experimental procedures. This index is highly valuable in QSPR analysis, which contributes to computer-assisted drug design \cite{ A.K2020, M.G2020}. Essentially, a topological index is a graph invariant that explains a molecular graph's structural features. Since Wiener's contribution \cite{H.Winener1947}, numerous types of indices have emerged in the literature relying on distinct parameters including degree, distance, eccentricity, spectrum, and so on.
Among various classes of topological indices, the vertex-degree-based topological indices are perhaps the most extensively studied (see \cite{I.Gutman2014, I.Gutman2018} and references therein).

All graphs considered in this paper are finite, simple and undirected. We refer to \cite{J.A.Bondy2008} for undefined notation and terminology. The vertex and edge sets of $G$ are denoted by $V(G)$ and $E(G)$, respectively. The {\it order} of $G$ (the number of its vertices) is $|V(G)|=n$ and its {\it size} of $G$ (the number of its edges) is $|E(G)|=m$. For a vertex $v\in V(G)$, the degree of $v$, denoted by $d_G(v)$, is the number of edges incident with $v$ in $G$. An edge connecting two adjacent vertices $u$ and $v$ is denoted as $uv\in E$. Let $e_{i,j}$ be the number of edges of i-degree vertices and j-degree vertices in $G$. Let $E_{i,j}$ be the set of edges of i-degree vertices and j-degree vertices in $G$. A vertex is said to a {\it pendant vertex} if its degree is one. Analogously, an edge is said to a {\it pendant edge} if it is incident with a pendant vertex.

For a vertex $v\in V(G)$, $G-v$ is the graph obtained from $G$ by deleting $v$ and its incident edges. If $S\subseteq E(G)$, we use $G-S$ to denote the graph formed from $G$ by removing the edges in $S$. Similarly, $G+S$ denotes the graph obtained from $G$ by adding the edges in $S$. In particular, if $S=\{uv\}$, then $G-S$ and $G+S$ are simply denoted by $G-uv$ and $G+uv$, respectively. If $M$ is a {\it matching}, the two ends of each edge of $M$ are said to be matched under $M$, and each vertex incident with an edge of $M$ is said to be covered by $M$. A {\it perfect matching} is one which covers every vertex of the graph.

Recall that an {\it acyclic graph} is one that contains no cycles. A connected acyclic graph is called a {\it tree}. A {\it molecular graph} is a graph with $d_G(u)\leq 4$ for all $u\in V(G)$. H. Y. Deng, Z. K. Tang and R. F. Wu \cite{H.Y.Deng2021} gave the sharp upper bound for the reduced Sombor index among all molecular trees of given order $n$. F. X. Wang and B. Wu \cite{F.X.Wang2022} gave the maximum value of the reduced Sombor index among all molecular trees of order $n$ with perfect matching and show that the maximum molecular trees of exponential reduced Sombor index. Du and Su \cite{J.W.Du2024} showed extremal results on bond incident degree indices of molecular trees with a fixed order and a fixed number of leaves. People may refer to \cite{S.Wagner2009, R.Cruz2021, F.X.Wang2024, A.Ali2022} for more relevant works.

Topological indices depending on end-vertex degrees of edges are called {\it vertex-degree-based topological indices} (VDB {\it topological indices} for short) \cite{Gutman2013}. During many years, scientists have been trying to improve the predictive power of the VDB topological index. The Sombor index and its variants belong to the class of vertex-degree-based (VDB) topological indices.
The Sombor index, a important index in chemical theory, is proposed in \cite{I.Gutman2021} and is defined as$$SO(G)=\sum_{uv\in E}\sqrt{d_u^2+d_v^2}.$$
Red\v{z}epovi\'{c} \cite{I2021} studied chemical applicability of Sombor index. Specifically, the Sombor index was used to model entropy and enthalpy of vaporization of alkanes with satisfactory prediction potential, indicating that this topological index may be used successfully on modeling thermodynamic properties of compounds.
The Sombor index has proven useful in QSPR and QSAR studies, contributing to its growing popularity in the literature \cite{I.Gutman2021, H.Liu2022, F.Movahedi2023}.
Recently, Rajathagiri introduced the Sombor index's variants in \cite{D.T.Rajathagiri2021}. F. Movahedi, I. Gutman, I. Red\u{z}epov\'{i}c and B. Furtula called the Sombor index's variants as Diminished Sombor index in \cite{F.Movahedi2026}. This index is defined as follows:$$DSO(G)=\sum_{uv\in E}\frac{\sqrt{d_u^2+d_v^2}}{d_u+d_v}.$$
 F. Movahedi, I. Gutman, I. Red\u{z}epov\'{i}c and B. Furtula \cite{F.Movahedi2026} derived bounds for the $DSO$ index and characterized the extremal graphs within the classes of trees, unicyclic graphs, and bicyclic graphs. They also presented numerical analyses regarding the structure-dependence of the $DSO$ index and its potential applications in chemistry. In \cite{F.Movahedi2025}, the tricyclic graph of a specified order that attains the maximum $DSO$ is identified, and its distinctive structural characteristics are examined. In \cite{F.Movahedi2025}, the relationships and inequalities between $DSO$ and some classical topological indices are analyzed thoroughly, with characterizations of extremal graphs achieving equality conditions. Motivated by known results, in this paper, we will characterize molecular trees with a perfect matching attaining the maximum $DSO$ index.

\section{Chemical applicability}
In order to evaluate the chemical relevance of the DSO index, we examined its statistical relationship with experimentally determined values for some physicochemical properties of octane isomers. The properties considered in this analysis are: density, melting point, boiling point, flash point, refractive index, critical volume, critical temperature, critical pressure, entropy. For each of the aforementioned properties, the Pearson correlation coefficient is computed to quantify the degree of linear association with the DSO index.

\begin{center}
\scalebox{0.4}{\includegraphics{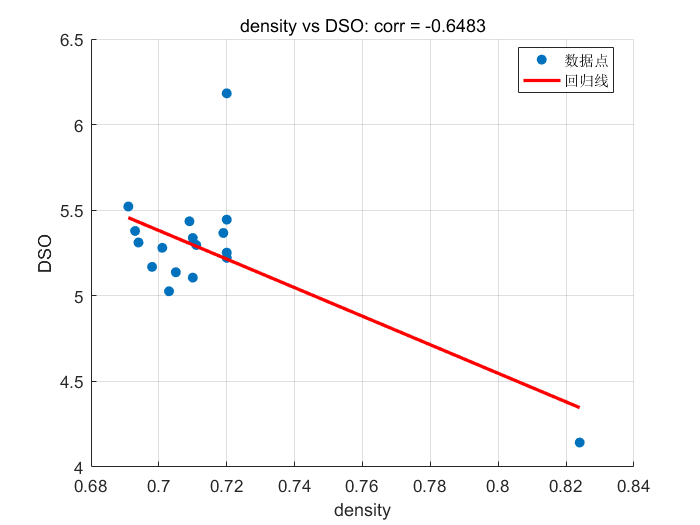}}\scalebox{0.4}{\includegraphics{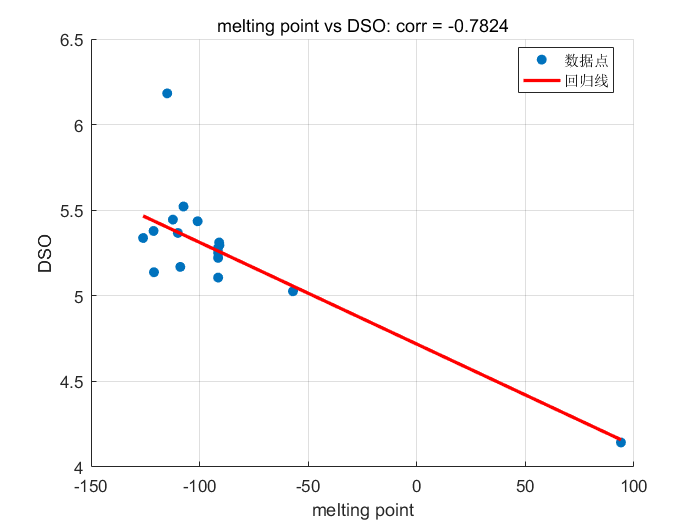}}\\
\scalebox{0.4}{\includegraphics{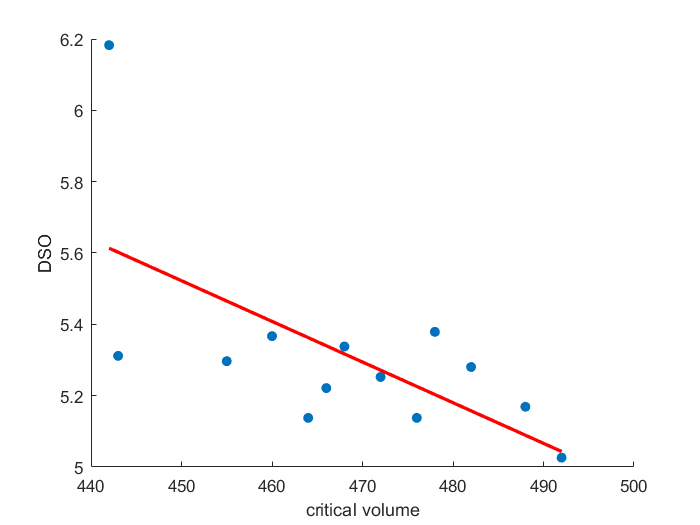}}\\
\vspace{0.1cm} \textbf{Figure 1.}
correlation between the DSO index and three properties.
\end{center}

In order to focus on the most statistically significant relationships, only those correlations with an absolute coefficient value $|corr| > 0.6$ were retained for further discussion. As illustrated in Figure 1, three properties satisfied this criterion: density, melting point and critical volume. These correlations suggest that the DSO index can predict some physicochemical properties of octane isomers.

\section{The maximum DSO index of Molecular Trees with a Perfect Matching}

In this section, we give the characterization for molecular trees which have maximum $DSO$ index among all molecular trees of order $n$ $(n\geq12)$ with a perfect matching. For molecular trees of $n\leq 10$ with a perfect matching, due to the order is small, it is easy to find the ones which has maximum $DSO$ index by simple calculation, see Figure 2.

\begin{center}
\scalebox{0.7}{\includegraphics{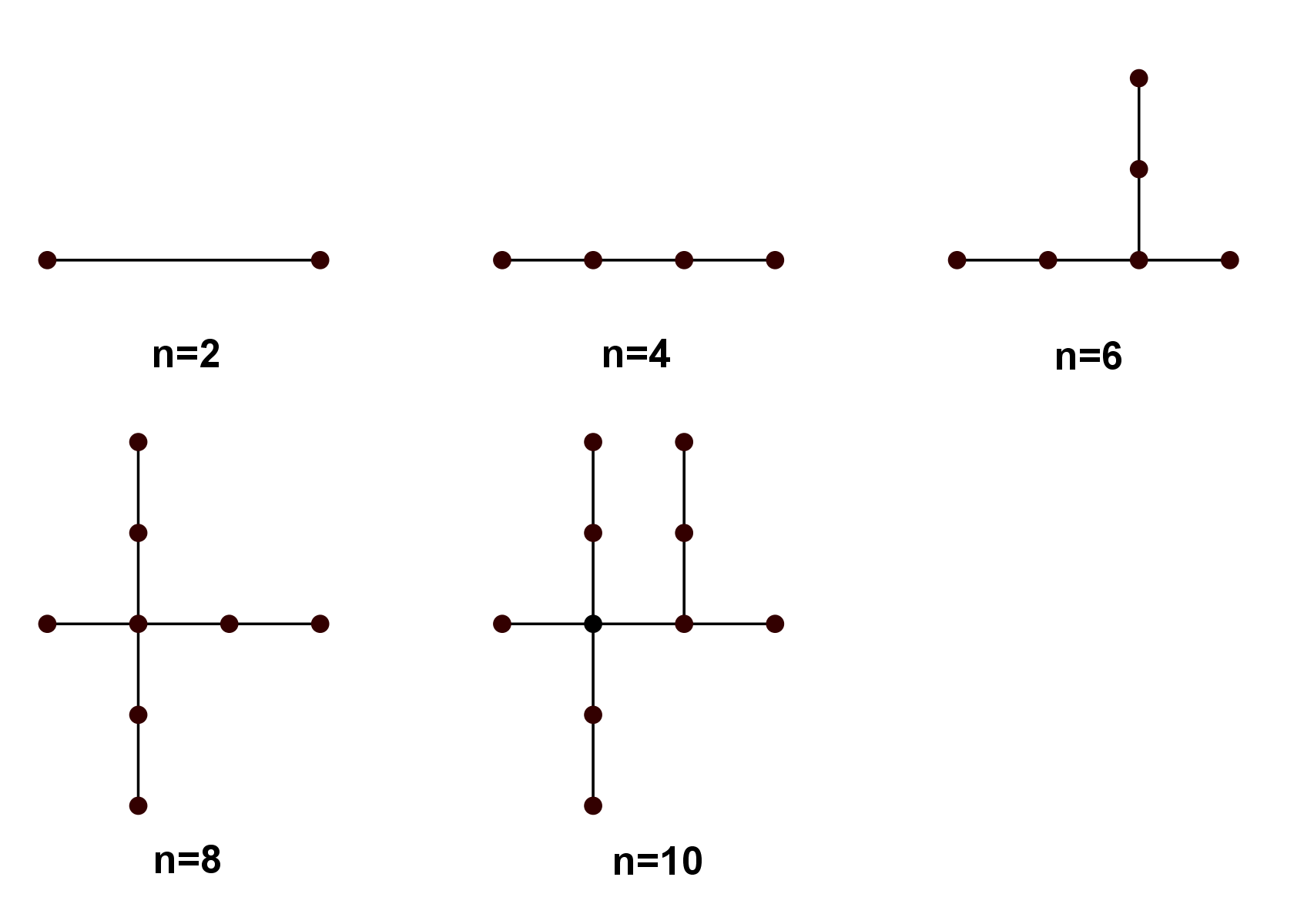}}\\
\vspace{0.1cm} \textbf{Figure 2.}
The corresponding molecular tree with the maximum $DSO$ index of order $n$ $(n\leq10)$.
\end{center}

Next, we give some useful lemmas and notations to describe properties of molecular trees which have maximum $DSO$ index among all molecular trees of order $n$ $(n\geq12)$ with a perfect matching.

\begin{lemma} \label{lem1}
Assume that a tree $T$ has the maximum $DSO$ index among all molecular trees of order $n$ $(n\geq12)$ with a perfect matching. Let $M$ be the perfect matching of $T$. If $e=uv\in M$, then $e$ is a pendant edge of $T$.
\end{lemma}

\begin{proof}

Suppose the result is not true. It follows that there must exist a vertex $u\in V(T)$ with $d_T(u)\geq 2$ such that for any $v\in N_T(u)$, $d_T(v)\geq 2$. We consider the following cases.

\vspace{2mm}\noindent{\bf Case 1.} $d_T(u)=2.$

Let $v, u_1\in N_T(u)$ and $uv\in M.$ Let $P$ be a maximal path which starts from $u_1$ and contains $u_1u$. Without loss of generality, suppose $x$ is another end-point of $P$. Obviously, $x$ is a pendant vertex. Let $y$ be the neighbor of $x$. Since $T$ has a perfect matching, then $d_T(y)=2$. Let $z$ be the another neighbor of $y$. Next, we distinguish the following two subcases.

\vspace{2mm}\noindent{\bf Case 1.1.} $z=v$.

Let $T^\prime=T-uu_1+yu_1$. Clearly, $T^\prime$ is also a molecular tree of order $n$ with a perfect matching. In the following, we will obtain a contradiction by showing $DSO(T^\prime)>DSO(T)$. Therefore
 \begin{flalign}
	&DSO(T')-DSO(T)\nonumber \\
&=\frac{\sqrt{d_{T^\prime}^2(u_1)+d_{T^\prime}^2(y)}}
{d_{T^\prime}(u_1)+d_{T^\prime}(y)}
-\frac{\sqrt{d_T^2(u_1)+d_T^2(u)}}{d_T(u_1)+d_T(u)}
+\frac{\sqrt{d_{T^\prime}^2(y)+d_{T^\prime}^2(x)}}
{d_{T^\prime}(y)+d_{T^\prime}(x)}\nonumber\\
&~~~~-\frac{\sqrt{d_T^2(y)+d_T^2(x)}}{d_T(y)+d_T(x)}
+\frac{\sqrt{d_{T^\prime}^2(y)+d_{T^\prime}^2(v)}}
{d_{T^\prime}(y)+d_{T^\prime}(v)}
-\frac{\sqrt{d_T^2(y)+d_T^2(v)}}
{d_T(y)+d_T(v)}\nonumber \\
&~~~~+\frac{\sqrt{d_{T^\prime}^2(u)+d_{T^\prime}^2(v)}}
{d_{T^\prime}(u)+d_{T^\prime}(v)}
-\frac{\sqrt{d_T^2(u)+d_T^2(v)}}{d_T(u)+d_T(v)}\nonumber\\
&=\frac{\sqrt{d_T^2(u_1)+9}}{d_T(u_1)+3}
-\frac{\sqrt{d_T^2(u_1)+4}}{d_T(u_1)+2}+\frac{\sqrt{10}}{4}-\frac{\sqrt{5}}{3}
+\frac{\sqrt{d_T^2(v)+9}}{d_T(v)+3}\nonumber\\
&~~~~-\frac{\sqrt{d_T^2(v)+4}}{d_T(v)+2}+\frac{\sqrt{d_T^2(v)+1}}{d_T(v)+1}
-\frac{\sqrt{d_T^2(v)+4}}{d_T(v)+2}.
    \end{flalign}

Since $2\leq d_T(u_1)\leq4$ and $2\leq d_T(v)\leq4$ in the last equation of $(1)$, a simple calculation shows that the value of $(1)$ become minimum when $d_T(u_1)=d_T(v)=4$. Thus $DSO(T^\prime)-DSO(T)=\frac{10}{7}+\frac{\sqrt{10}}{4}-\frac{4\sqrt{5}}{3}+\frac{\sqrt{17}}{5}>0$,
contradicting the maximality of $T$.

\vspace{2mm}\noindent{\bf Case 1.2.} $z\neq v$.

Let $T^\prime=T-uu_1+yu_1$. Thus

       	\begin{flalign}
&DSO(T^\prime)-DSO(T)\nonumber\\
&=\frac{\sqrt{d_{T^\prime}^2(u_1)+d_{T^\prime}^2(y)}}
{d_{T^\prime}(u_1)+d_{T^\prime}(y)}
-\frac{\sqrt{d_T^2(u_1)+d_T^2(u)}}{d_T(u_1)+d_T(u)}
+\frac{\sqrt{d_{T^\prime}^2(y)+d_{T^\prime}^2(x)}}
{d_{T^\prime}(y)+d_{T^\prime}(x)}\nonumber\\
&~~~~-\frac{\sqrt{d_T^2(y)+d_T^2(x)}}{d_T(y)+d_T(x)}
+\frac{\sqrt{d_{T^\prime}^2(y)+d_{T^\prime}^2(z)}}
{d_{T^\prime}(y)+d_{T^\prime}(z)}
-\frac{\sqrt{d_T^2(y)+d_T^2(z)}}{d_T(y)+d_T(z)}\nonumber\\
&~~~~+\frac{\sqrt{d_{T^\prime}^2(u)+d_{T^\prime}^2(v)}}
{d_{T^\prime}(u)+d_{T^\prime}(v)}
-\frac{\sqrt{d_T^2(u)+d_T^2(v)}}{d_T(u)+d_T(v)}\nonumber\\
&=\frac{\sqrt{d_T^2(u_1)+9}}{d_T(u_1)+3}
-\frac{\sqrt{d_T^2(u_1)+4}}{d_T(u_1)+2}+\frac{\sqrt{10}}{4}-\frac{\sqrt{5}}{3}
+\frac{\sqrt{d_T^2(z)+9}}{d_T(z)+3}\nonumber\\
&~~~~-\frac{\sqrt{d_T^2(z)+4}}{d_T(z)+2}+\frac{\sqrt{d_T^2(v)+1}}{d_T(v)+1}
-\frac{\sqrt{d_T^2(v)+4}}{d_T(v)+2}.
\end{flalign}

Note that $2\leq d_T(u_1)\leq4$, $3\leq d_T(v)\leq4$ and $2\leq d_T(z)\leq4$. By checking through all possibilities, we can find the right side of $(2)$ achieve its minimum value at $d_T(u_1)=d_T(z)=4$ and $d_T(v)=2$. Thus $DSO(T^\prime)-DSO(T)=\frac{10}{7}+\frac{\sqrt{10}}{4}-\frac{2\sqrt{5}}{3}+\frac{\sqrt{2}}{2}>0$,
contradicting the maximality of $T$.

\vspace{2mm}\noindent{\bf Case 2.} $d_T(u)=3.$

Let $v, u_1, u_2\in N_T(u)$ and $uv\in M$. Let $d_T(v)\geq3$. Otherwise, by Case 1, $v$ is adjacent to a pendant vertex, a contradiction to $uv\in M.$ Let $P$ be a maximal path which starts from $u$ and contains $uv$. Without loss of generality, suppose $x$ is another end-point of $P$. Obviously, $x$ is a pendant vertex. Let $y$ be the neighbor of $x$. Since $T$ has a perfect matching, then $d_T(y)=2.$ Let $z$ be another neighbor of $y$. Next, we distinguish the following two subcases.

\vspace{2mm}\noindent{\bf Case 2.1.} $z=v.$

Let$T^\prime=T-\{uu_1,uu_2\}+\{yu_1,yu_2\}$. Since $d_T(v)\geq3$, we have

       	\begin{flalign}
&DSO(T^\prime)-DSO(T)\nonumber\\
&=\frac{\sqrt{d_{T^\prime}^2(u_1)+d_{T^\prime}^2(y)}}
{d_{T^\prime}(u_1)+d_{T^\prime}(y)}
-\frac{\sqrt{d_T^2(u_1)+d_T^2(u)}}{d_T(u_1)+d_T(u)}
+\frac{\sqrt{d_{T^\prime}^2(y)+d_{T^\prime}^2(u_2)}}
{d_{T^\prime}(y)+d_{T^\prime}(u_2)}\nonumber\\
&~~~~-\frac{\sqrt{d_T^2(u_2)+d_T^2(u)}}{d_T(u_2)+d_T(u)}
+\frac{\sqrt{d_{T^\prime}^2(u)+d_{T^\prime}^2(v)}}
{d_{T^\prime}(u)+d_{T^\prime}(v)}
-\frac{\sqrt{d_T^2(u)+d_T^2(v)}}{d_T(u)+d_T(v)}\nonumber \\
&~~~~+\frac{\sqrt{d_{T^\prime}^2(y)+d_{T^\prime}^2(x)}}
{d_{T^\prime}(y)+d_{T^\prime}(x)}
-\frac{\sqrt{d_T^2(y)+d_T^2(x)}}{d_T(y)+d_T(x)}
+\frac{\sqrt{d_{T^\prime}^2(v)+d_{T^\prime}^2(y)}}
{d_{T^\prime}(v)+d_{T^\prime}(y)}\nonumber\\
&~~~~-\frac{\sqrt{d_T^2(y)+d_T^2(v)}}{d_T(y)+d_T(v)}\nonumber\\
&=\frac{\sqrt{d_T^2(u_1)+16}}{d_T(u_1)+4}
-\frac{\sqrt{d_T^2(u_1)+9}}{d_T(u_1)+3}+\frac{\sqrt{17}}{5}-\frac{\sqrt{5}}{3}
+\frac{\sqrt{d_T^2(v)+16}}{d_T(v)+4}\nonumber\\
&~~~~-\frac{\sqrt{d_T^2(v)+4}}{d_T(v)+2}+\frac{\sqrt{d_T^2(v)+1}}{d_T(v)+1}
-\frac{\sqrt{d_T^2(v)+9}}{d_T(v)+3}+\frac{\sqrt{d_T^2(u_2)+16}}{d_T(u_2)+4}\nonumber\\
&~~~~-\frac{\sqrt{d_T^2(u_2)+9}}{d_T(u_2)+3}.
\end{flalign}

Note that $2\leq d_T(u_1)\leq4$, $3\leq d_T(v)\leq4$ and $2\leq d_T(u_2)\leq4$. By checking through all possibilities, we can find the right side of $(3)$ achieve its minimum value at $d_T(u_1)=d_T(u_2)=4$ and $d_T(v)=2$. Thus $DSO(T^\prime)-DSO(T)=\frac{3\sqrt{2}}{2}+\frac{\sqrt{17}}{5}
-\frac{10}{7}-\frac{\sqrt{5}}{3}-\frac{\sqrt{13}}{5}>0$,
contradicting the maximality of $T$.

\vspace{2mm}\noindent{\bf Case 2.2.} $z\neq v.$

Let $T^\prime=T-\{uu_1,uu_2\}+\{yu_1,yu_2\}$. Since $d_T(v)\geq3$, we have

       	\begin{flalign}
&DSO(T^\prime)-DSO(T)\nonumber\\
&=\frac{\sqrt{d_{T^\prime}^2(u_1)+d_{T^\prime}^2(y)}}
{d_{T^\prime}(u_1)+d_{T^\prime}(y)}
-\frac{\sqrt{d_T^2(u_1)+d_T^2(u)}}{d_T(u_1)+d_T(u)}
+\frac{\sqrt{d_{T^\prime}^2(y)+d_{T^\prime}^2(u_2)}}
{d_{T^\prime}(y)+d_{T^\prime}(u_2)}\nonumber\\
&~~~~-\frac{\sqrt{d_T^2(u)+d_T^2(u_2)}}{d_T(u)+d_T(u_2)}
+\frac{\sqrt{d_{T^\prime}^2(u)+d_{T^\prime}^2(v)}}
{d_{T^\prime}(u)+d_{T^\prime}(v)}
-\frac{\sqrt{d_T^2(u)+d_T^2(v)}}{d_T(u)+d_T(v)}\nonumber\\
&~~~~+\frac{\sqrt{d_{T^\prime}^2(y)+d_{T^\prime}^2(x)}}
{d_{T^\prime}(y)+d_{T^\prime}(x)}
-\frac{\sqrt{d_T^2(y)+d_T^2(x)}}{d_T(y)+d_T(x)}
+\frac{\sqrt{d_{T^\prime}^2(y)+d_{T^\prime}^2(z)}}
{d_{T^\prime}(y)+d_{T^\prime}(z)}\nonumber\\
&~~~~-\frac{\sqrt{d_T^2(y)+d_T^2(z)}}{d_T(y)+d_T(z)}\nonumber\\
&=\frac{\sqrt{d_T^2(u_1)+16}}{d_T(u_1)+4}
-\frac{\sqrt{d_T^2(u_1)+9}}{d_T(u_1)+3}+\frac{\sqrt{d_T^2(u_2)+16}}{d_T(u_2)+4}
-\frac{\sqrt{d_T^2(u_2)+9}}{d_T(u_2)+3}\nonumber\\
&~~~~+\frac{\sqrt{17}}{5}-\frac{\sqrt{5}}{3}
+\frac{\sqrt{d_T^2(z)+16}}{d_T(z)+4}-\frac{\sqrt{d_T^2(z)+4}}{d_T(z)+2}
+\frac{\sqrt{d_T^2(v)+1}}{d_T(v)+1}\nonumber\\
&~~~~-\frac{\sqrt{d_T^2(v)+9}}{d_T(v)+3}.
\end{flalign}

Note that $2\leq d_T(u_1)\leq4$, $3\leq d_T(v)\leq4$, $2\leq d_T(u_2)\leq4$ and $2\leq d_T(z)\leq4$. By checking through all possibilities, we can find the right side of $(4)$ achieve its minimum value at $d_T(u_1)=d_T(u_2)=d_T(z)=4$ and $d_T(v)=3$. Thus $DSO(T^\prime)-DSO(T)=\frac{\sqrt{10}+4\sqrt{2}}{4}-\frac{10}{7}
+\frac{4\sqrt{17}}{5}-\frac{2\sqrt{5}}{3}>0$,
contradicting the maximality of $T$.

\vspace{2mm}\noindent{\bf Case 3.} $d_T(u)=4.$

Let $v, u_1, u_2, u_3\in N_T(u)$ and $uv\in M.$ Let $d_T(v)=4$. Otherwise, by Case 2, $v$ is adjacent to a pendant vertex, a contradiction to $uv\in M.$ Let $P$ be a maximal path which starts from $u$ and contains $uv.$ Without loss of generality, suppose $x_1$ is another end-point of $P$. Obviously, $x_1$ is a pendant vertex. Let $y_1$ be the neighbor of $x_1.$ Since $T$ has a perfect matching, we have $d_T(y_1)=2.$ Let $z_1$ be another neighbor of $y_1.$ Let $P_1$ be a maximal path which starts from $u$ and contains $uu_1.$ Without loss of generality, suppose $x_2$ is another end-point of $P_1$. Obviously, $x_2$ is a pendant vertex. Let $y_2$ be the neighbor of $x_2.$ Since $T$ has a perfect matching, we have $d_T(y_2)=2.$ Let $z_2$ be another neighbor of $y_2.$ Next, we distinguish the following six subcases.

\vspace{2mm}\noindent{\bf Case 3.1.} $z_1=v, y_2=u_1.$

Let $T^\prime=T-\{uu_1,uu_2,uu_3\}+\{y_1u_1,y_1u_3,u_1u_3\}$. Since $d_T(u_1)=d_T(y_2)=2$ and $d_T(v)=4$ , we have

       	\begin{flalign}
    &DSO(T^\prime)-DSO(T)\nonumber\\
    &=\frac{\sqrt{d_{T^\prime}^2(u_1)+d_{T^\prime}^2(y_1)}}
{d_{T^\prime}(u_1)+d_{T^\prime}(y_1)}
-\frac{\sqrt{d_T^2(u_1)+d_T^2(u)}}{d_T(u_1)+d_T(u)}
+\frac{\sqrt{d_{T^\prime}^2(u_2)+d_{T^\prime}^2(u_1)}}
{d_{T^\prime}(u_2)+d_{T^\prime}(u_1)}\nonumber\\
&~~~~-\frac{\sqrt{d_T^2(u)+d_T^2(u_2)}}{d_T(u)+d_T(u_2)}
+\frac{\sqrt{d_{T^\prime}^2(y_1)+d_{T^\prime}^2(u_3)}}
{d_{T^\prime}(y_1)+d_{T^\prime}(u_3)}
-\frac{\sqrt{d_T^2(u_3)+d_T^2(u)}}{d_T(u_3)+d_T(u)}\nonumber \\
&~~~~+\frac{\sqrt{d_{T^\prime}^2(u_1)+d_{T^\prime}^2(x_2)}}
{d_{T^\prime}(u_1)+d_{T^\prime}(x_2)}
-\frac{\sqrt{d_T^2(u_1)+d_T^2(x_2)}}{d_T(u_1)+d_T(x_2)}
+\frac{\sqrt{d_{T^\prime}^2(u)+d_{T^\prime}^2(v)}}
{d_{T^\prime}(u)+d_{T^\prime}(v)}\nonumber\\
&~~~~-\frac{\sqrt{d_T^2(u)+d_T^2(v)}}{d_T(u)+d_T(v)}+\frac{\sqrt{d_{T^\prime}^2(y_1)+d_{T^\prime}^2(x_1)}}
{d_{T^\prime}(y_1)+d_{T^\prime}(x_1)}-\frac{\sqrt{d_T^2(y_1)+d_T^2(x_1)}}{d_T(y_1)+d_T(x_1)}\nonumber\\
&~~~~+\frac{\sqrt{d_{T^\prime}^2(v)+d_{T^\prime}^2(y_1)}}
{d_{T^\prime}(v)+d_{T^\prime}(y_1)}
-\frac{\sqrt{d_T^2(v)+d_T^2(y_1)}}{d_T(v)+d_T(y_1)}\nonumber\\
&=\frac{\sqrt{d_T^2(u_2)+9}}{d_T(u_2)+3}
-\frac{\sqrt{d_T^2(u_2)+16}}{d_T(u_2)+4}+\frac{\sqrt{10}}{4}
-\sqrt{5}+\frac{\sqrt{13}}{5}\nonumber\\
&~~~-\frac{\sqrt{2}}{2}
+\frac{2\sqrt{17}}{5}.
\end{flalign}

Since $2\leq d_T(u_2)\leq 4$, by taking each possible value 2, 3, 4 into the last equation of (5), we find that the right side of $(5)$ achieve its minimum at $d_T(u_2)=2$. Thus $DSO(T^\prime)-DSO(T)=\frac{\sqrt{13}+2\sqrt{17}}{5}-\frac{5\sqrt{5}}{3}
+\frac{\sqrt{10}}{4}+\frac{5}{7}>0$,
contradicting the maximality of $T$.

\vspace{2mm}\noindent{\bf Case 3.2.} $z_1=v, z_2=v.$

Let $T^\prime=T-\{uu_1,uu_2,uu_3\}+\{y_1u_1,y_1u_3,y_2u_2\}$. Since $d_T(v)=4$ , we have

       	\begin{flalign}
&DSO(T^\prime)-DSO(T)\nonumber\\
&=\frac{\sqrt{d_{T^\prime}^2(y_2)+d_{T^\prime}^2(u_2)}}
{d_{T^\prime}(y_2)+d_{T^\prime}(u_2)}
-\frac{\sqrt{d_T^2(u_2)+d_T^2(u)}}{d_T(u_2)+d_T(u)}
+\frac{\sqrt{d_{T^\prime}^2(y_1)+d_{T^\prime}^2(u_1)}}
{d_{T^\prime}(y_1)+d_{T^\prime}(u_1)}\nonumber\\
&~~~~-\frac{\sqrt{d_T^2(u)+d_T^2(u_1)}}{d_T(u)+d_T(u_1)}
+\frac{\sqrt{d_{T^\prime}^2(y_1)+d_{T^\prime}^2(u_3)}}
{d_{T^\prime}(y_1)+d_{T^\prime}(u_3)}
-\frac{\sqrt{d_T^2(u_3)+d_T^2(u)}}{d_T(u_3)+d_T(u)}\nonumber \\
&~~~~+\frac{\sqrt{d_{T^\prime}^2(y_2)+d_{T^\prime}^2(x_2)}}
{d_{T^\prime}(y_2)+d_{T^\prime}(x_2)}
-\frac{\sqrt{d_T^2(y_2)+d_T^2(x_2)}}{d_T(y_2)+d_T(x_2)}
+\frac{\sqrt{d_{T^\prime}^2(u_1)+d_{T^\prime}^2(y_2)}}
{d_{T^\prime}(u_1)+d_{T^\prime}(y_2)}\nonumber\\
&~~~~-\frac{\sqrt{d_T^2(u_1)+d_T^2(y_2)}}{d_T(u_1)+d_T(y_2)}
+\frac{\sqrt{d_{T^\prime}^2(y_1)+d_{T^\prime}^2(x_1)}}
{d_{T^\prime}(y_1)+d_{T^\prime}(x_1)}
-\frac{\sqrt{d_T^2(y_1)+d_T^2(x_1)}}{d_T(y_1)+d_T(x_1)}\nonumber\\
&~~~~+\frac{\sqrt{d_{T^\prime}^2(v)+d_{T^\prime}^2(y_1)}}
{d_{T^\prime}(v)+d_{T^\prime}(y_1)}
-\frac{\sqrt{d_T^2(v)+d_T^2(y_1)}}{d_T(v)+d_T(y_1)}
+\frac{\sqrt{d_{T^\prime}^2(v)+d_{T^\prime}^2(u)}}
{d_{T^\prime}(v)+d_{T^\prime}(u)}\nonumber\\
&~~~~-\frac{\sqrt{d_T^2(v)+d_T^2(u)}}{d_T(v)+d_T(u)}\nonumber\\
&=\frac{\sqrt{d_T^2(u_2)+9}}{d_T(u_2)+3}
-\frac{\sqrt{d_T^2(u_2)+16}}{d_T(u_2)+4}+\frac{\sqrt{d_T^2(u_1)+9}}{d_T(u_1)+3}
-\frac{\sqrt{d_T^2(u_1)+4}}{d_T(u_1)+2}\nonumber\\
&~~~~+\frac{\sqrt{10}}{4}+\frac{2\sqrt{17}}{5}-\sqrt{5}.
\end{flalign}

Since $2\leq d_T(u_2)\leq4$ and $2\leq d_T(u_1)\leq4$ in the last equation of $(6)$, a simple calculation shows that the value of $(6)$ become minimum when $d_T(u_1)=4$ and $d_T(u_2)=2$.
Thus $DSO(T^\prime)-DSO(T)=\frac{\sqrt{13}+2\sqrt{17}}{5}-\frac{5\sqrt{5}}{3}
+\frac{\sqrt{10}}{4}+\frac{5}{7}>0$,
contradicting the maximality of $T$.

\vspace{2mm}\noindent{\bf Case 3.3.} $z_1=v_2, z_2\neq v.$

Let $T^\prime=T-\{uv,uv_3,uv_4\}+\{y_1v,y_1v_4,y_2v_3\}$. Since $d_T(v_2)=4$ , we have

       	\begin{flalign}
&DSO(T^\prime)-DSO(T)\nonumber\\
&=\frac{\sqrt{d_{T^\prime}^2(y_2)+d_{T^\prime}^2(u_2)}}
{d_{T^\prime}(y_2)+d_{T^\prime}(u_2)}
-\frac{\sqrt{d_T^2(u_2)+d_T^2(u)}}{d_T(u_2)+d_T(u)}
+\frac{\sqrt{d_{T^\prime}^2(y_1)+d_{T^\prime}^2(u_1)}}
{d_{T^\prime}(y_1)+d_{T^\prime}(u_1)}\nonumber\\
&~~~~-\frac{\sqrt{d_T^2(u)+d_T^2(u_1)}}{d_T(u)+d_T(u_1)}
+\frac{\sqrt{d_{T^\prime}^2(y_1)+d_{T^\prime}^2(u_3)}}
{d_{T^\prime}(y_1)+d_{T^\prime}(u_3)}
-\frac{\sqrt{d_T^2(u_3)+d_T^2(u)}}{d_T(u_3)+d_T(u)}\nonumber \\
&~~~~+\frac{\sqrt{d_{T^\prime}^2(y_2)+d_{T^\prime}^2(x_2)}}
{d_{T^\prime}(y_2)+d_{T^\prime}(x_2)}
-\frac{\sqrt{d_T^2(y_2)+d_T^2(x_2)}}{d_T(y_2)+d_T(x_2)}
+\frac{\sqrt{d_{T^\prime}^2(z_2)+d_{T^\prime}^2(y_2)}}
{d_{T^\prime}(z_2)+d_{T^\prime}(y_2)}\nonumber\\
&~~~~-\frac{\sqrt{d_T^2(y_2)+d_T^2(z_2)}}{d_T(y_2)+d_T(z_2)}
+\frac{\sqrt{d_{T^\prime}^2(x_1)+d_{T^\prime}^2(y_1)}}
{d_{T^\prime}(x_1)+d_{T^\prime}(y_1)}
-\frac{\sqrt{d_T^2(x_1)+d_T^2(y_1)}}{d_T(x_1)+d_T(y_1)}\nonumber\\
&~~~~+\frac{\sqrt{d_{T^\prime}^2(y_1)+d_{T^\prime}^2(v)}}
{d_{T^\prime}(y_1)+d_{T^\prime}(v)}
-\frac{\sqrt{d_T^2(y_1)+d_T^2(v)}}{d_T(y_1)+d_T(v)}
+\frac{\sqrt{d_{T^\prime}^2(u)+d_{T^\prime}^2(v)}}
{d_{T^\prime}(u)+d_{T^\prime}(v)}\nonumber\\
&~~~~-\frac{\sqrt{d_T^2(u)+d_T^2(v)}}{d_T(u)+d_T(v)}
+\frac{\sqrt{d_{T^\prime}^2(u_1)+d_{T^\prime}^2(z_2)}}
{d_{T^\prime}(u_1)+d_{T^\prime}(z_2)}
-\frac{\sqrt{d_T^2(u_1)+d_T^2(z_2)}}{d_T(u_1)+d_T(z_2)}\nonumber\\
&=\frac{\sqrt{d_T^2(u_2)+9}}{d_T(u_2)+3}
-\frac{\sqrt{d_T^2(u_2)+16}}{d_T(u_2)+4}+\frac{\sqrt{d_T^2(z_2)+9}}{d_T(z_2)+3}
-\frac{\sqrt{d_T^2(z_2)+4}}{d_T(z_2)+2}\nonumber\\
&~~~~+\frac{\sqrt{10}}{4}+\frac{2\sqrt{17}}{5}-\sqrt{5}.
\end{flalign}

Note that $2\leq d_T(u_2)\leq4$ and $3\leq d_T(z_2)\leq4$. By checking through all possibilities, we can find the right side of $(7)$ achieve its minimum value at $d_T(u_2)=2$ and $d_T(z_2)=4$. Thus
Thus $DSO(T^\prime)-DSO(T)=\frac{\sqrt{13}+2\sqrt{17}}{5}-\frac{5\sqrt{5}}{3}
+\frac{\sqrt{10}}{4}+\frac{5}{7}>0$,
contradicting the maximality of $T$.

\vspace{2mm}\noindent{\bf Case 3.4.} $z_1\neq v, y_2=u_1.$

Let $T^\prime=T-\{uu_1,uu_2,uu_3\}+\{y_1u_1,y_1u_3,u_1u_2\}$. Since $d_T(u_1)=d_T(y_2)=2$ and $d_T(v)=4$ , we have

       	\begin{flalign}
&DSO(T^\prime)-DSO(T)\nonumber\\
&=\frac{\sqrt{d_{T^\prime}^2(u_1)+d_{T^\prime}^2(u_2)}}
{d_{T^\prime}(u_1)+d_{T^\prime}(u_2)}
-\frac{\sqrt{d_T^2(u_2)+d_T^2(u)}}{d_T(u_2)+d_T(u)}
+\frac{\sqrt{d_{T^\prime}^2(u_1)+d_{T^\prime}^2(y_1)}}
{d_{T^\prime}(u_1)+d_{T^\prime}(y_1)}\nonumber\\
&~~~~-\frac{\sqrt{d_T^2(u)+d_T^2(u_1)}}{d_T(u)+d_T(u_1)}
+\frac{\sqrt{d_{T^\prime}^2(u_3)+d_{T^\prime}^2(y_1)}}
{d_{T^\prime}(y_1)+d_{T^\prime}(u_3)}
-\frac{\sqrt{d_T^2(u)+d_T^2(u_3)}}{d_T(u)+d_T(u_3)}\nonumber \\
&~~~~+\frac{\sqrt{d_{T^\prime}^2(u_1)+d_{T^\prime}^2(x_2)}}
{d_{T^\prime}(u_1)+d_{T^\prime}(x_2)}
-\frac{\sqrt{d_T^2(u_1)+d_T^2(x_2)}}{d_T(u_1)+d_T(x_2)}
+\frac{\sqrt{d_{T^\prime}^2(x_1)+d_{T^\prime}^2(y_1)}}
{d_{T^\prime}(x_1)+d_{T^\prime}(y_1)}\nonumber\\
&~~~~-\frac{\sqrt{d_T^2(x_1)+d_T^2(y_1)}}{d_T(x_1)+d_T(y_1)}
+\frac{\sqrt{d_{T^\prime}^2(y_1)+d_{T^\prime}^2(z_1)}}
{d_{T^\prime}(y_1)+d_{T^\prime}(z_1)}
-\frac{\sqrt{d_T^2(y_1)+d_T^2(z_1)}}{d_T(y_1)+d_T(z_1)}\nonumber\\
&~~~~+\frac{\sqrt{d_{T^\prime}^2(v)+d_{T^\prime}^2(u)}}
{d_{T^\prime}(v)+d_{T^\prime}(u)}
-\frac{\sqrt{d_T^2(u)+d_T^2(v)}}{d_T(u)+d_T(v)}\nonumber\\
&=\frac{\sqrt{d_T^2(u_2)+9}}{d_T(u_2)+3}
-\frac{\sqrt{d_T^2(u_2)+16}}{d_T(u_2)+4}
+\frac{\sqrt{d_T^2(z_1)+16}}{d_T(z_1)+4}
-\frac{\sqrt{d_T^2(z_1)+4}}{d_T(z_1)+2}\nonumber\\
&~~~~+\frac{5}{7}-\sqrt{5}
+\frac{\sqrt{10}}{4}-\frac{2\sqrt{17}}{5}
-\frac{\sqrt{2}}{2}.
\end{flalign}

Note that $2\leq d_T(u_2)\leq4$ and $2\leq d_T(z_1)\leq4$. By checking through all possibilities, we can find the right side of $(8)$ achieve its minimum value at $d_T(u_2)=2$ and $d_T(z_1)=4$.
Thus $DSO(T^\prime)-DSO(T)=\frac{\sqrt{13}+2\sqrt{17}}{5}-\frac{5\sqrt{5}}{3}
+\frac{\sqrt{10}}{4}+\frac{5}{7}>0$,
contradicting the maximality of $T$.

\vspace{2mm}\noindent{\bf Case 3.5.} $z_1\neq v, z_2=u_1.$

Let $T^\prime=T-\{uu_1,uu_2,uu_3\}+\{y_1u_1,y_1u_3,y_2u_2\}$. Since $d_T(v)=4$ , we have

       	\begin{flalign}
&DSO(T^\prime)-DSO(T)\nonumber\\
&=\frac{\sqrt{d_{T^\prime}^2(u_1)+d_{T^\prime}^2(y_1)}}
{d_{T^\prime}(u_1)+d_{T^\prime}(y_1)}
-\frac{\sqrt{d_T^2(u_1)+d_T^2(u)}}{d_T(u_1)+d_T(u)}
+\frac{\sqrt{d_{T^\prime}^2(y_2)+d_{T^\prime}^2(u_2)}}
{d_{T^\prime}(y_2)+d_{T^\prime}(u_2)}\nonumber\\
&~~~~-\frac{\sqrt{d_T^2(u)+d_T^2(u_2)}}{d_T(u)+d_T(u_2)}
+\frac{\sqrt{d_{T^\prime}^2(y_1)+d_{T^\prime}^2(u_3)}}
{d_{T^\prime}(y_1)+d_{T^\prime}(u_3)}
-\frac{\sqrt{d_T^2(u_3)+d_T^2(u)}}{d_T(u_3)+d_T(u)}\nonumber \\
&~~~~+\frac{\sqrt{d_{T^\prime}^2(y_2)+d_{T^\prime}^2(x_2)}}
{d_{T^\prime}(y_2)+d_{T^\prime}(x_2)}
-\frac{\sqrt{d_T^2(y_2)+d_T^2(x_2)}}{d_T(y_2)+d_T(x_2)}
+\frac{\sqrt{d_{T^\prime}^2(u_1)+d_{T^\prime}^2(y_2)}}
{d_{T^\prime}(u_1)+d_{T^\prime}(y_2)}\nonumber\\
&~~~~-\frac{\sqrt{d_T^2(u_1)+d_T^2(y_2)}}{d_T(u_1)+d_T(y_2)}
+\frac{\sqrt{d_{T^\prime}^2(y_1)+d_{T^\prime}^2(x_1)}}
{d_{T^\prime}(y_1)+d_{T^\prime}(x_1)}
-\frac{\sqrt{d_T^2(y_1)+d_T^2(x_1)}}{d_T(y_1)+d_T(x_1)}\nonumber\\
&~~~~+\frac{\sqrt{d_{T^\prime}^2(y_1)+d_{T^\prime}^2(z_1)}}
{d_{T^\prime}(y_1)+d_{T^\prime}(z_1)}
-\frac{\sqrt{d_T^2(y_1)+d_T^2(z_1)}}{d_T(y_1)+d_T(z_1)}
+\frac{\sqrt{d_{T^\prime}^2(v)+d_{T^\prime}^2(u)}}
{d_{T^\prime}(v)+d_{T^\prime}(u)}\nonumber\\
&~~~~-\frac{\sqrt{d_T^2(v)+d_T^2(u)}}{d_T(v)+d_T(u)}\nonumber\\
&=\frac{\sqrt{d_T^2(u_2)+9}}{d_T(u_2)+3}
-\frac{\sqrt{d_T^2(u_2)+16}}{d_T(u_2)+4}
+\frac{\sqrt{d_T^2(u_1)+9}}{d_T(u_1)+3}+\frac{\sqrt{10}}{4}
+\frac{2\sqrt{17}}{5}\nonumber\\
&~~~~-\frac{\sqrt{d_T^2(u_1)+4}}{d_T(u_1)+2}
+\frac{\sqrt{d_T^2(z_1)+16}}{d_T(z_1)+4}
-\frac{\sqrt{d_T^2(z_1)+4}}{d_T(z_1)+2}-\frac{2\sqrt{5}}{3}
-\frac{\sqrt{2}}{2}.
\end{flalign}

Since $2\leq d_T(u_2)\leq4$, $2\leq d_T(u_1)\leq4$ and $2\leq d_T(z_1)\leq4$ in the last equation of $(9)$, a simple calculation shows that the value of $(9)$ become minimum when $d_T(u_1)=d_T(z_1)=4$ and $d_T(u_2)=2$.
Thus $DSO(T^\prime)-DSO(T)=\frac{\sqrt{13}+2\sqrt{17}}{5}-\frac{5\sqrt{5}}{3}
+\frac{\sqrt{10}}{4}+\frac{5}{7}>0$,
contradicting the maximality of $T$.

\vspace{2mm}\noindent{\bf Case 3.6.} $z_1\neq v, z_2\neq u_1.$

Let $T^\prime=T-\{uu_1,uu_2,uu_3\}+\{y_1u_1,y_1u_3,y_2u_2\}$. Since $d_T(v_2)=4$ , we have

       	\begin{flalign}
&DSO(T^\prime)-DSO(T)\nonumber\\
&=\frac{\sqrt{d_{T^\prime}^2(u_2)+d_{T^\prime}^2(y_2)}}
{d_{T^\prime}(u_2)+d_{T^\prime}(y_2)}
-\frac{\sqrt{d_T^2(u_2)+d_T^2(u)}}{d_T(u_2)+d_T(u)}
+\frac{\sqrt{d_{T^\prime}^2(y_1)+d_{T^\prime}^2(u_1)}}
{d_{T^\prime}(y_1)+d_{T^\prime}(u_1)}\nonumber\\
&~~~~-\frac{\sqrt{d_T^2(u)+d_T^2(u_1)}}{d_T(u)+d_T(u_1)}
+\frac{\sqrt{d_{T^\prime}^2(y_1)+d_{T^\prime}^2(u_3)}}
{d_{T^\prime}(y_1)+d_{T^\prime}(u_3)}
-\frac{\sqrt{d_T^2(u_3)+d_T^2(u)}}{d_T(u_3)+d_T(u)}\nonumber \\
&~~~~+\frac{\sqrt{d_{T^\prime}^2(y_2)+d_{T^\prime}^2(x_2)}}
{d_{T^\prime}(y_2)+d_{T^\prime}(x_2)}
-\frac{\sqrt{d_T^2(y_2)+d_T^2(x_2)}}{d_T(y_2)+d_T(x_2)}
+\frac{\sqrt{d_{T^\prime}^2(z_2)+d_{T^\prime}^2(y_2)}}
{d_{T^\prime}(z_2)+d_{T^\prime}(y_2)}\nonumber\\
&~~~~-\frac{\sqrt{d_T^2(z_2)+d_T^2(y_2)}}{d_T(z_2)+d_T(y_2)}
+\frac{\sqrt{d_{T^\prime}^2(y_1)+d_{T^\prime}^2(x_1)}}
{d_{T^\prime}(y_1)+d_{T^\prime}(x_1)}
-\frac{\sqrt{d_T^2(y_1)+d_T^2(x_1)}}{d_T(y_1)+d_T(x_1)}\nonumber\\
&~~~~+\frac{\sqrt{d_{T^\prime}^2(y_1)+d_{T^\prime}^2(z_1)}}
{d_{T^\prime}(y_1)+d_{T^\prime}(z_1)}
-\frac{\sqrt{d_T^2(y_1)+d_T^2(z_1)}}{d_T(y_1)+d_T(z_1)}
+\frac{\sqrt{d_{T^\prime}^2(v)+d_{T^\prime}^2(u)}}
{d_{T^\prime}(v)+d_{T^\prime}(u)}\nonumber\\
&~~~~-\frac{\sqrt{d_T^2(v)+d_T^2(u)}}{d_T(v)+d_T(u)}\nonumber\\
&=\frac{\sqrt{d_T^2(u_2)+9}}{d_T(u_2)+3}
-\frac{\sqrt{d_T^2(u_2)+16}}{d_T(u_2)+4}
+\frac{\sqrt{d_T^2(z_2)+9}}{d_T(z_2)+3}+\frac{\sqrt{10}}{4}
+\frac{2\sqrt{17}}{5}\nonumber\\
&~~~~-\frac{\sqrt{d_T^2(z_2)+4}}{d_T(z_2)+2}
+\frac{\sqrt{d_T^2(z_1)+16}}{d_T(z_1)+4}
-\frac{\sqrt{d_T^2(z_1)+4}}{d_T(z_1)+2}-\frac{2\sqrt{5}}{3}
-\frac{\sqrt{2}}{2}.
\end{flalign}

Note that $2\leq d_T(u_2)\leq4$, $2\leq d_T(z_1)\leq4$ and $2\leq d_T(u_2)\leq4$. By checking through all possibilities, we can find the right side of $(10)$ achieve its minimum value at $d_T(z_2)=d_T(z_1)=4$ and $d_T(u_2)=2$.
Thus $DSO(T^\prime)-DSO(T)=\frac{\sqrt{13}+2\sqrt{17}}{5}-\frac{5\sqrt{5}}{3}
+\frac{\sqrt{10}}{4}+\frac{5}{7}>0$,
contradicting the maximality of $T$.

The proof is completed.
\end{proof}

\begin{lemma}\label{lem2}
Let $T$ be a molecular tree which has maximum $DOS$ index among all molecular trees of order $n$ $(n\geq12)$ with a perfect matching $M$. If $d_T(u)=3$, then $d_T(v)\in \{1, 4\}$ for any $v\in N_T(u)$.
\end{lemma}

\begin{proof} Suppose that $d_T(v)\in\{2, 3\}$. We consider two cases and will obtain a contradiction.

\vspace{2mm}\noindent{\bf Case 1.}  $d_T(v)=3$.

Let $N_T(u)\setminus \{v\}=\{u_1, u_2\}$ and $N_T(v)\setminus \{u\}=\{v_1, v_2\}$. Since $T$ has a perfect matching, by Lemma 1 $u$ is matched with one of its neighbor, say $u_1$, $v$ is matched with one of its neighbor, say $v_1$. Thus $d_T(u_1)=1$, $d_T(v_1)=1$, $2\leq d_T(u_2)\leq 4$ and $2\leq d_T(v_2)\leq4$.
Let $T^\prime=T-uu_2+vu_2$. Clearly, $T^\prime$ is also a molecular tree of order $n$ with a perfect matching. In the following, we will obtain a contradiction by showing $DSO(T^\prime)>DSO(T)$.

\begin{flalign}
&DSO(T^\prime)-DSO(T)\nonumber\\
&=\frac{\sqrt{d_{T^\prime}^2(u_2)+d_{T^\prime}^2(v)}}
{d_{T^\prime}(u_2)+d_{T^\prime}(v)}
-\frac{\sqrt{d_T^2(u_2)+d_T^2(u)}}{d_T(u_2)+d_T(u)}
+\frac{\sqrt{d_{T^\prime}^2(u)+d_{T^\prime}^2(u_1)}}
{d_{T^\prime}(u)+d_{T^\prime}(u_1)}\nonumber\\
&~~~~-\frac{\sqrt{d_T^2(u)+d_T^2(u_1)}}{d_T(u)+d_T(u_1)}
+\frac{\sqrt{d_{T^\prime}^2(u)+d_{T^\prime}^2(v)}}
{d_{T^\prime}(u)+d_{T^\prime}(v)}
-\frac{\sqrt{d_T^2(u)+d_T^2(v)}}{d_T(u)+d_T(v)}\nonumber \\
&~~~~+\frac{\sqrt{d_{T^\prime}^2(v)+d_{T^\prime}^2(v_1)}}
{d_{T^\prime}(v)+d_{T^\prime}(v_1)}
-\frac{\sqrt{d_T^2(v)+d_T^2(v_1)}}{d_T(v)+d_T(v_1)}
+\frac{\sqrt{d_{T^\prime}^2(v)+d_{T^\prime}^2(v_2)}}
{d_{T^\prime}(v)+d_{T^\prime}(v_2)}\nonumber\\
&~~~~-\frac{\sqrt{d_T^2(v)+d_T^2(v_2)}}{d_T(v)+d_T(v_2)}\nonumber\\
&=\frac{\sqrt{d_T^2(u_2)+16}}{d_T(u_2)+4}
-\frac{\sqrt{d_T^2(u_2)+9}}{d_T(u_2)+3}
+\frac{\sqrt{d_T^2(v_2)+16}}{d_T(v_2)+4}\nonumber\\
&~~~~-\frac{\sqrt{d_T^2(v_2)+9}}{d_T(v_2)+3}
+\frac{2\sqrt{5}}{3}
-\frac{\sqrt{10}+\sqrt{2}}{2}+\frac{\sqrt{17}}{5}.
\end{flalign}

Since $2\leq d_T(u_2)\leq4$ and $2\leq d_T(v_2)\leq4$ in the last equation of $(11)$, a simple calculation shows that the value of $(11)$ become minimum when $d_T(u_2)=d_T(v_2)=4$.
Thus $DSO(T^\prime)-DSO(T)=\frac{\sqrt{2}-\sqrt{10}}{2}-\frac{10}{7}
+\frac{\sqrt{17}}{5}+\frac{2\sqrt{5}}{3}>0$,
contradicting the maximality of $T$.

\vspace{2mm}\noindent{\bf Case 2.} $d_T(v)=2$

Let $N_T(u)\setminus \{v\}=\{u_1, u_2\}$ and $N_T(v)\setminus \{u\}=\{v_1\}$. Since $T$ has a perfect matching, by Lemma 1 $u$ is matched with one of its neighbor, say $u_1$, $v$ is matched with one of its neighbor, say $v_1$. Thus $d_T(u_1)=1$, $d_T(v_1)=1$. In view of Case 1 and the assumption that $n\geq 12$, we have $d_T(u_2)=4$.
Let $P$ be a maximal path of $T$ which starts from $u$ containing $uu_2$. Without loss of generality, suppose $x$ is another end-point of P. Obviously, $x$ is a pendant vertex. Let $y$ be the neighbor of $x$. Since $T$ has a perfect matching, we have $d_T(y)=2$. Let $z$ be another neighbor of $y$. Since $n\geq 12$, we have $u_2\neq z$.

\vspace{2mm}\noindent{\bf Case 2.1.} $d_T(z)=3.$

Let $N_T(z)\setminus \{y\}=\{z_1, z_2\}$. Since $T$ has a perfect matching, by Lemma 1 $z$ is matched with one of its neighbor, say $z_1$. Thus $d_T(z_1)=1$. In view of Case 1 and $n\geq 12$, we can assume that $d_T(z_2)=4$.
Let $T^\prime=T-uv+zv$. Clearly, $T^\prime$ is also a molecular tree of order $n$ with a perfect matching. In the following, we will obtain a contradiction by showing $DSO(T^\prime)>DSO(T)$.

\begin{flalign}
&DSO(T^\prime)-DSO(T)\nonumber\\
&=\frac{\sqrt{d_{T^\prime}^2(v)+d_{T^\prime}^2(z)}}
{d_{T^\prime}(v)+d_{T^\prime}(z)}
-\frac{\sqrt{d_T^2(v)+d_T^2(u)}}{d_T(v)+d_T(u)}
+\frac{\sqrt{d_{T^\prime}^2(u)+d_{T^\prime}^2(u_1)}}
{d_{T^\prime}(u)+d_{T^\prime}(u_1)}\nonumber\\
&~~~~-\frac{\sqrt{d_T^2(u)+d_T^2(u_1)}}{d_T(u)+d_T(u_1)}
+\frac{\sqrt{d_{T^\prime}^2(u)+d_{T^\prime}^2(u_2)}}
{d_{T^\prime}(u)+d_{T^\prime}(u_2)}
-\frac{\sqrt{d_T^2(u)+d_T^2(u_2)}}{d_T(u)+d_T(u_2)}\nonumber \\
&~~~~+\frac{\sqrt{d_{T^\prime}^2(z)+d_{T^\prime}^2(z_1)}}
{d_{T^\prime}(z)+d_{T^\prime}(z_1)}
-\frac{\sqrt{d_T^2(z)+d_T^2(z_1)}}{d_T(z)+d_T(z_1)}
+\frac{\sqrt{d_{T^\prime}^2(z)+d_{T^\prime}^2(y)}}
{d_{T^\prime}(z)+d_{T^\prime}(y)}\nonumber\\
&~~~~-\frac{\sqrt{d_T^2(z)+d_T^2(y)}}{d_T(z)+d_T(y)}
+\frac{\sqrt{d_{T^\prime}^2(z_2)+d_{T^\prime}^2(z)}}
{d_{T^\prime}(z_2)+d_{T^\prime}(z)}
-\frac{\sqrt{d_T^2(z_2)+d_T^2(z)}}{d_T(z_2)+d_T(z)}\nonumber\\
&=\frac{4\sqrt{5}}{3}
+\frac{\sqrt{17}-2\sqrt{13}}{5}
-\frac{10}{7}+\frac{\sqrt{2}-\sqrt{10}}{2}>0.\nonumber\\
\end{flalign}

\vspace{2mm}\noindent{\bf Case 2.2} $d_T(z)=4.$

Let $N_T(z)\setminus \{y\}=\{z_1, z_2, z_3\}$ Since $T$ has a perfect matching, by Lemma 1 $z$ is matched with one of its neighbor, say $z_1$. Thus $d_T(z_1)=1$ and $3\leq d_T(z_2)\leq4$. By the choice of $P$, we have $z_3=2$.

\vspace{2mm}\noindent{\bf Case 2.2.1.} $d_T(z)=4, d_T(z_2)=3.$

Let $N_T(z_2)\setminus \{z\}=\{z_4, z_5\}$. Since $T$ has a perfect matching, by Lemma 1 $z_2$ is matched with one of its neighbor, say $z_4$. Thus $d_T(z_4)=1$. In view of Case 1, we may assume that $3 \leq d_T(z_5)\leq 4$.
Let $T^\prime=T-uv+vz_2$. Clearly, $T^\prime$ is a molecular tree of order $n$ with a perfect matching. In the following, we will obtain a contradiction by showing $DSO(T^\prime)>DSO(T)$.

\begin{flalign}
&DSO(T^\prime)-DSO(T)\nonumber\\
&=\frac{\sqrt{d_{T^\prime}^2(z_2)+d_{T^\prime}^2(v)}}
{d_{T^\prime}(z_2)+d_{T^\prime}(v)}
-\frac{\sqrt{d_T^2(v)+d_T^2(u)}}{d_T(v)+d_T(u)}
+\frac{\sqrt{d_{T^\prime}^2(u)+d_{T^\prime}^2(u_1)}}
{d_{T^\prime}(u)+d_{T^\prime}(u_1)}\nonumber\\
&~~~~-\frac{\sqrt{d_T^2(u)+d_T^2(u_1)}}{d_T(u)+d_T(u_1)}
+\frac{\sqrt{d_{T^\prime}^2(u)+d_{T^\prime}^2(u_2)}}
{d_{T^\prime}(u)+d_{T^\prime}(u_2)}
-\frac{\sqrt{d_T^2(u)+d_T^2(u_2)}}{d_T(u)+d_T(u_2)}\nonumber \\
&~~~~+\frac{\sqrt{d_{T^\prime}^2(z)+d_{T^\prime}^2(z_2)}}
{d_{T^\prime}(z)+d_{T^\prime}(z_2)}
-\frac{\sqrt{d_T^2(z)+d_T^2(z_2)}}{d_T(z)+d_T(z_2)}
+\frac{\sqrt{d_{T^\prime}^2(z_2)+d_{T^\prime}^2(z_4)}}
{d_{T^\prime}(z_2)+d_{T^\prime}(z_4)}\nonumber\\
&~~~~-\frac{\sqrt{d_T^2(z_2)+d_T^2(z_4)}}{d_T(z_2)+d_T(z_4)}
+\frac{\sqrt{d_{T^\prime}^2(z_2)+d_{T^\prime}^2(z_5)}}
{d_{T^\prime}(z_2)+d_{T^\prime}(z_5)}
-\frac{\sqrt{d_T^2(z_2)+d_T^2(z_5)}}{d_T(z_2)+d_T(z_5)}\nonumber\\
&=\frac{\sqrt{d_T^2(z_5)+16}}{d_T(z_5)+4}
-\frac{\sqrt{d_T^2(z_5)+9}}{d_T(z_5)+3}
+\sqrt{5}-\frac{10}{7}+\frac{\sqrt{2}-\sqrt{10}}{2}\nonumber\\
&~~~~+\frac{\sqrt{17}-\sqrt{13}}{5}.
\end{flalign}

Since $3\leq d_T(z_5)\leq 4$, by taking each possible value 3, 4 into the last equation of (12), we find that the right side of $(12)$ achieve its minimum at $d_T(z_5)=4$. Thus $DSO(T^\prime)-DSO(T)=\frac{2\sqrt{2}-\sqrt{10}}{2}-\frac{15}{7}
+\sqrt{5}+\frac{\sqrt{17}-\sqrt{13}}{5}>0$,
contradicting the maximality of $T$.

\vspace{2mm}\noindent{\bf Case 2.2.2.} $d_T(z)=4, d_T(z_2)=4.$

Let $N_T(z_3)\setminus \{z\}=\{z_4\}$. Since $T$ has a perfect matching, by Lemma 1 $z_3$ is matched with one of its neighbor, say $z_4$. Thus $d_T(z_4)=1$.
Let $T^\prime=T-\{zz_3,zy\}+\{vz_3,vy\}$. Clearly, $T^\prime$ is also a molecular tree of order $n$ with a perfect matching. In the following, we will obtain a contradiction by showing $DSO(T^\prime)>DSO(T)$.

\begin{flalign}
&DSO(T^\prime)-DSO(T)\nonumber\\
&=\frac{\sqrt{d_{T^\prime}^2(v)+d_{T^\prime}^2(y)}}
{d_{T^\prime}(v)+d_{T^\prime}(y)}
-\frac{\sqrt{d_T^2(z)+d_T^2(y)}}{d_T(z)+d_T(y)}
+\frac{\sqrt{d_{T^\prime}^2(v)+d_{T^\prime}^2(z_3)}}
{d_{T^\prime}(v)+d_{T^\prime}(z_3)}\nonumber\\
&~~~~-\frac{\sqrt{d_T^2(z)+d_T^2(z_3)}}{d_T(z)+d_T(z_3)}
+\frac{\sqrt{d_{T^\prime}^2(v)+d_{T^\prime}^2(v_1)}}
{d_{T^\prime}(v)+d_{T^\prime}(v_1)}
-\frac{\sqrt{d_T^2(v)+d_T^2(v_1)}}{d_T(v)+d_T(v_1)}\nonumber \\
&~~~~+\frac{\sqrt{d_{T^\prime}^2(v)+d_{T^\prime}^2(u)}}
{d_{T^\prime}(v)+d_{T^\prime}(u)}
-\frac{\sqrt{d_T^2(v)+d_T^2(u)}}{d_T(v)+d_T(u)}
+\frac{\sqrt{d_{T^\prime}^2(z)+d_{T^\prime}^2(z_1)}}
{d_{T^\prime}(z)+d_{T^\prime}(z_1)}\nonumber\\
&~~~~-\frac{\sqrt{d_T^2(z)+d_T^2(z_1)}}{d_T(z)+d_T(z_1)}
+\frac{\sqrt{d_{T^\prime}^2(z)+d_{T^\prime}^2(z_2)}}
{d_{T^\prime}(z)+d_{T^\prime}(z_2)}
-\frac{\sqrt{d_T^2(z)+d_T^2(z_2)}}{d_T(z)+d_T(z_2)}\nonumber\\
&=\frac{\sqrt{5}}{3}-\frac{\sqrt{2}}{2}
+\frac{5}{7}-\frac{\sqrt{13}}{5}>0.\nonumber
\end{flalign}

The proof is completed.
\end{proof}

\begin{lemma}\label{lem3}
If a molecular tree $T$ has the maximum $DOS$ index among all molecular trees of order $n$ $(n\geq12)$ with a perfect matching $M$, then $e_{4,4}\leq2$ in $T$.
\end{lemma}

\begin{proof}

We prove by contradiction. Suppose that $e_{4,4}\geq3$. Let $u_1u_2, u_3u_4, u_5u_6\in E_{4,4}$. Let $P$ be a maximal path in $T$. Let $x$ be an end of $P$. Obviously, $x$ is a pendant vertex. Let $y$ be the neighbor of $x$. Since $T$ has a perfect matching, $d_T(y)=2$. Let $z$ be the neighbor of $y$ other than $x$. Let $N_T(z)\setminus \{y\}=\{z_1, z_2, z_3\}$.
Since $T$ has a perfect matching, by Lemma 1 $z$ is matched with one of its neighbor, say $z_1$. Thus $d_T(z_1)=1$.
Since $P$ is the maximal path, Lemma 1 and Lemma 2,  we have $d_T(z_2)=2$ and $3\leq d_T(z_3)\leq4$.

\vspace{2mm}\noindent{\bf Case 1.} $d_T(z_3)=3.$

Let $N_T(z_3)\setminus \{z\}=\{z_4, z_5\}$. Since $T$ has a perfect matching, by Lemma 1 $z_3$ is matched with one of its neighbor, say $z_4$. Thus $d_T(z_4)=1$. In view of Case 1 of Lemma 2, we may assume that $d_T(z_5)=4$. Let $T^\prime=T-\{u_1u_2,u_3u_4,u_5u_6,zy,zz_2,zz_3\}+\{u_1z,u_2z,u_3y,u_4y,u_5z_2,u_6z_2\}.$ Therefore,

\begin{flalign}
&DSO(T^\prime)-DSO(T)\nonumber\\
&=\frac{\sqrt{d_{T^\prime}^2(u_1)+d_{T^\prime}^2(z)}}
{d_{T^\prime}(u_1)+d_{T^\prime}(z)}
+\frac{\sqrt{d_{T^\prime}^2(u_2)+d_{T^\prime}^2(z)}}
{d_{T^\prime}(u_2)+d_{T^\prime}(z)}
+\frac{\sqrt{d_{T^\prime}^2(u_3)+d_{T^\prime}^2(y)}}
{d_{T^\prime}(u_3)+d_{T^\prime}(y)}\nonumber\\
&~~~~+\frac{\sqrt{d_{T^\prime}^2(y)+d_{T^\prime}^2(u_4)}}
{d_{T^\prime}(y)+d_{T^\prime}(u_4)}
+\frac{\sqrt{d_{T^\prime}^2(u_5)+d_{T^\prime}^2(z_2)}}
{d_{T^\prime}(u_5)+d_{T^\prime}(z_2)}
+\frac{\sqrt{d_{T^\prime}^2(u_6)+d_{T^\prime}^2(z_2)}}
{d_{T^\prime}(u_6)+d_{T^\prime}(z_2)}\nonumber \\
&~~~~-\frac{\sqrt{d_T^2(u_1)+d_T^2(u_2)}}{d_T(u_1)+d_T(u_2)}
-\frac{\sqrt{d_T^2(u_3)+d_T^2(u_4)}}{d_T(u_3)+d_T(u_4)}
-\frac{\sqrt{d_T^2(u_5)+d_T^2(u_6)}}{d_T(u_5)+d_T(u_6)}\nonumber\\
&~~~~-\frac{\sqrt{d_T^2(z)+d_T^2(y)}}{d_T(z)+d_T(y)}
-\frac{\sqrt{d_T^2(z)+d_T^2(z_2)}}{d_T(z)+d_T(z_2)}
-\frac{\sqrt{d_T^2(z)+d_T^2(z_3)}}{d_T(z)+d_T(z_3)}\nonumber\\
&~~~~+\frac{\sqrt{d_{T^\prime}^2(z)+d_{T^\prime}^2(z_1)}}
{d_{T^\prime}(z)+d_{T^\prime}(z_1)}
-\frac{\sqrt{d_T^2(z)+d_T^2(z_1)}}{d_T(z)+d_T(z_1)}
+\frac{\sqrt{d_{T^\prime}^2(z_3)+d_{T^\prime}^2(z_4)}}
{d_{T^\prime}(z_3)+d_{T^\prime}(z_4)}\nonumber\\
&~~~~-\frac{\sqrt{d_T^2(z_3)+d_T^2(z_4)}}{d_T(z_3)+d_T(z_4)}
+\frac{\sqrt{d_{T^\prime}^2(z_3)+d_{T^\prime}^2(z_5)}}
{d_{T^\prime}(z_3)+d_{T^\prime}(z_5)}
-\frac{\sqrt{d_T^2(z_3)+d_T^2(z_5)}}{d_T(z_3)+d_T(z_5)}\nonumber\\
&=\frac{\sqrt{20}}{7}-\frac{2\sqrt{5}}{3}
+\frac{\sqrt{10}-3\sqrt{2}}{2}-\frac{\sqrt{17}}{5}>0.\nonumber
\end{flalign}

\vspace{2mm}\noindent{\bf Case 2.} $d_T(z_3)=4.$

Let $N_T(z_3)\setminus \{z\}=\{z_4, z_5, z_6\}$. Since $T$ has a perfect matching, by Lemma 1 $z_3$ is matched with one of its neighbor, say $z_4$. Thus $d_T(z_4)=1$. $T^\prime=T-\{u_1u_2,u_3u_4,zy,zz_2\}+\{u_1z_2,u_2z_2,u_3y,u_4y\}.$

\begin{flalign}
&DSO(T^\prime)-DSO(T)\nonumber\\
&=\frac{\sqrt{d_{T^\prime}^2(u_1)+d_{T^\prime}^2(z_2)}}
{d_{T^\prime}(u_1)+d_{T^\prime}(z_2)}
+\frac{\sqrt{d_{T^\prime}^2(u_2)+d_{T^\prime}^2(z_2)}}
{d_{T^\prime}(u_2)+d_{T^\prime}(z_2)}
+\frac{\sqrt{d_{T^\prime}^2(u_3)+d_{T^\prime}^2(y)}}
{d_{T^\prime}(u_3)+d_{T^\prime}(y)}\nonumber\\
&~~~~+\frac{\sqrt{d_{T^\prime}^2(y)+d_{T^\prime}^2(u_4)}}
{d_{T^\prime}(y)+d_{T^\prime}(u_4)}
-\frac{\sqrt{d_T^2(u_1)+d_T^2(u_2)}}{d_T(u_1)+d_T(u_2)}
-\frac{\sqrt{d_T^2(u_3)+d_T^2(u_4)}}{d_T(u_3)+d_T(u_4)}\nonumber\\
&~~~~-\frac{\sqrt{d_T^2(z)+d_T^2(y)}}{d_T(z)+d_T(y)}
-\frac{\sqrt{d_T^2(z)+d_T^2(z_2)}}{d_T(z)+d_T(z_2)}
+\frac{\sqrt{d_{T^\prime}^2(z)+d_{T^\prime}^2(z_1)}}
{d_{T^\prime}(z)+d_{T^\prime}(z_1)}\nonumber \\
&~~~~-\frac{\sqrt{d_T^2(z)+d_T^2(z_1)}}{d_T(z)+d_T(z_1)}
+\frac{\sqrt{d_{T^\prime}^2(z_3)+d_{T^\prime}^2(z)}}
{d_{T^\prime}(z_3)+d_{T^\prime}(z)}
-\frac{\sqrt{d_T^2(z_3)+d_T^2(z)}}{d_T(z_3)+d_T(z)}\nonumber\\
&=\frac{\sqrt{20}}{7}-\frac{2\sqrt{5}}{3}
+\frac{\sqrt{10}-3\sqrt{2}}{2}-\frac{\sqrt{17}}{5}>0.\nonumber
\end{flalign}

This complete the proof of this lemma.
\end{proof}

Before presenting our main theorem, we define several classes of trees.
Let $\mathcal{T}_{1, 3}$ be the set of trees in which all vertices have degree 1 or 3.
One can see that for any tree $T\in \mathcal{T}_{1, 3}$ of order $n$, the number of pendant vertices is $\frac n 2+1$ and the number of vertices of degree 3 is $\frac n 2-1$. Three types of trees, denoted by $\mathcal{H}_i (i\in\{0, 1, 2\})$, are further obtained by $T\in \mathcal{T}_{1, 3}$ inserting some vertices as follows:

$\mathcal{H}_0=\{T: T$ is a tree obtained from $T'\in \mathcal{T}_{1, 3}$ inserting exactly one vertex into each edge with both ends having degree $3\}$.

$\mathcal{H}_1=\{T: T$ is a tree obtained from $T'\in \mathcal{T}_{1, 3}$ inserting exactly one vertex into all but one edge with both ends having degree $3\}$.

$\mathcal{H}_2=\{T: T$ is a tree obtained from $T'\in \mathcal{T}_{1, 3}$ inserting exactly one vertex into all edges but two with both ends having degree $3\}$.
By the above definitions, if $n$ is the order of a tree $T$, then
\resizebox{0.4\textwidth}{!}{
    $
    \begin{array}{l}

       	\begin{aligned}
      n=\left \{
\begin{array}{ll}
3k+7, &\mbox{ $T\in \mathcal{H}_0$ \ }\\
3k+6, & \mbox{ $T\in \mathcal{H}_1$\
}\\
3k+8, &\mbox{ $T\in \mathcal{H}_2$, \ }
\end{array}
\right.

       \end{aligned}
     \end{array}
$
}

where $k$ is a nonnegative integer. Furthermore, let
$\mathcal{G}_i=\{T: T$ is a tree obtained $T'\in \mathcal{H}_i$ by adding $V(T')$ new vertices each of which joined to one vertex of $T'\}$ for each $i\in \{0, 1, 2\}$. Figure 3 shows some element in $\mathcal{G}_0, \mathcal{G}_1, \mathcal{G}_2$, respectiely.

\begin{center}
\scalebox{0.6}{\includegraphics{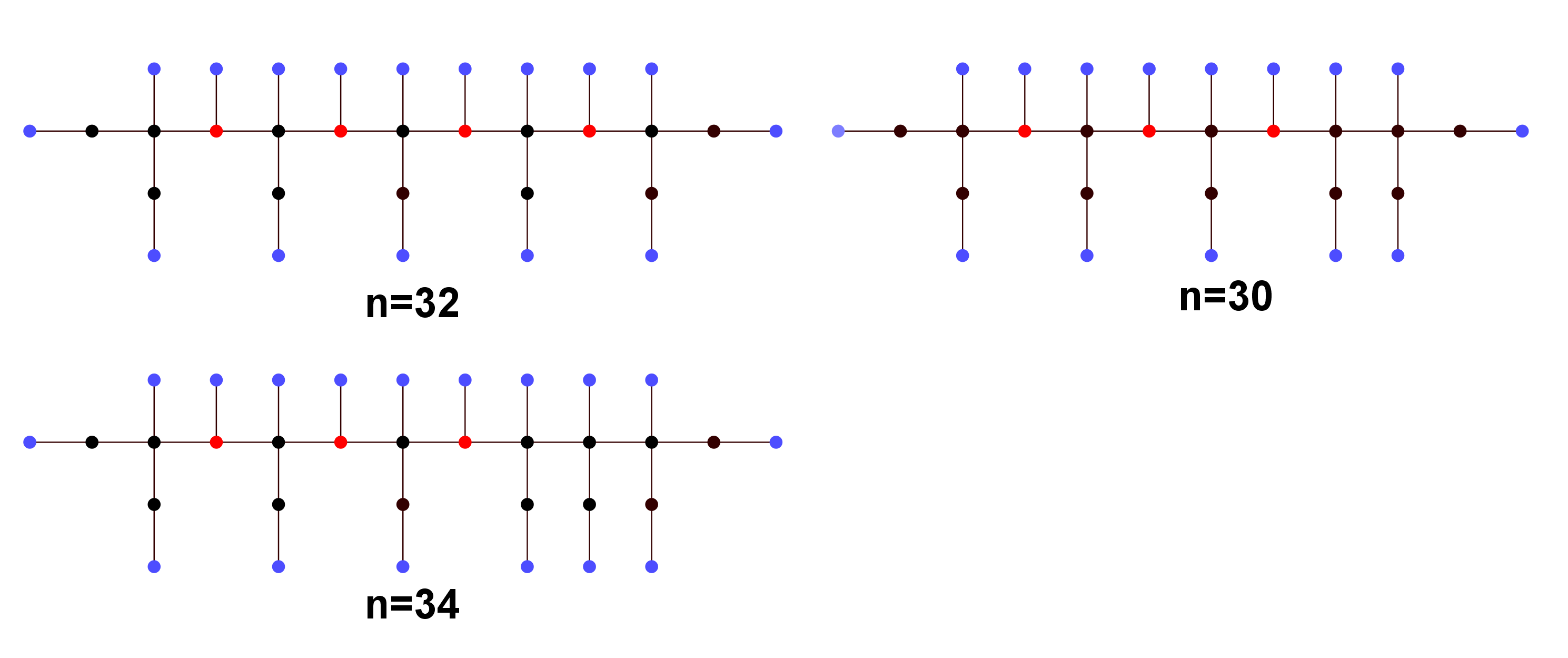}}\\
\vspace{0.1cm} \textbf{Figure 3.}
The tree with the maximum $DSO$ index of molecular tree of order $n$ $(n\geq12)$ with a perfect matching.
\end{center}

\begin{theorem}
If $T$ is a molecular tree of order $n$ $(n\geq12)$ with perfect matching $M$, then

\resizebox{1.1\textwidth}{!}{
    $
    \begin{array}{l}

       	\begin{aligned}
      DSO(T)\leq\left \{
\begin{array}{ll}
(\frac{\sqrt{17}}{30}+\frac{\sqrt{5}}{9}+\frac{\sqrt{10}}{24}\frac{5}{21})n
-\frac{\sqrt{17}}{15}+\frac{10\sqrt{5}}{9}-\frac{10}{3}-\frac{40}{21}, &\mbox{ $n-12 \equiv 2 (mod\ 6)$ \ }\\
(\frac{\sqrt{17}}{30}+\frac{\sqrt{5}}{9}+\frac{\sqrt{10}}{24}\frac{5}{21})n
+\frac{4\sqrt{5}}{3}-\frac{\sqrt{10}}{2}-\frac{20}{7}-\frac{\sqrt{2}}{2}, & \mbox{ $n-12 \equiv 0 (mod\ 6)$\
}\\
(\frac{\sqrt{17}}{30}+\frac{\sqrt{5}}{9}+\frac{\sqrt{10}}{24}\frac{5}{21})n
+\frac{\sqrt{17}}{15}+\frac{14\sqrt{5}}{9}-\frac{4\sqrt{10}}{6}-\frac{80}{21} -\sqrt{2}, &\mbox{ $n-12 \equiv 4 (mod\ 6)$. \ }
\end{array}
\right.

       \end{aligned}
     \end{array}
$
}

For $n-12 \equiv 2 (mod\ 6)$, the equality holds if and only if $T\in \mathcal{G}_0$; for $n-12 \equiv 0 (mod\ 6)$, the equality holds if and only if $T\in \mathcal{G}_1$; for $n-12 \equiv 4 (mod\ 6)$, the equality holds if and only if $T\in \mathcal{G}_2$.
\end{theorem}

\begin{proof}
Let $T^*$ be a molecular tree which has maximum $DOS$ index with a perfect matching $M$. Let $n_i$ be the number of vertices of degree $i$ in $T^*$ for each $i\in\{1, 2, 3, 4\}$. Clearly,
 \begin{flalign}
&n_1+n_2+n_3+n_4=n
\end{flalign}
and
\begin{flalign}
&n_1+2n_2+3n_3+4n_4=2m=2(n-1).
\end{flalign}

By Lemma 1 and Lemma 2, we have
\begin{flalign}
&e_{1,2}+e_{1,3}+e_{1,4}+e_{2,4}+e_{3,4}+e_{4,4}=n-1\nonumber\\
&~~~~~~~~~2n_2+3n_3+n_4+e_{4,4}=n-1.
\end{flalign}

By Lemma 1, we have $n_1=\frac{n}{2}.$
Combining (14), (15) and (16), we have $n_3=-2n_4+\frac{n}{2}-2, n_2=n_4+2, e_{4,4}=3n_4-\frac{n}{2}+1.$

Thus
\begin{flalign}
&DSO(T)=\sum_{uv\in E_{1,2}}f(uv)+\sum_{uv\in E_{1,3}}f(uv)+\sum_{uv\in E_{1,4}}f(uv)+\sum_{uv\in E_{2,4}}f(uv)\nonumber\\
&~~~~~~~~~~~~~~+\sum_{uv\in E_{3,4}}f(uv)+\sum_{uv\in E_{4,4}}f(uv),\nonumber
\end{flalign}
where $f(uv)=\frac{\sqrt{d_G^2(u)+d_G^2(v)}}{d_G(u)+d_G(v)}$.
We consider three cases in terms of $n$.

\vspace{2mm}\noindent{\bf Case 1.} $n=6k+14,$ for an integer $k\geq0.$

Combining $n=6k+14,$ for an integer $k\geq0,$ we have $e_{4,4}=3(n_4-k)-6.$ By Lemma 3, we have $e_{4,4}=0,1,2$. When $e_{4,4}=1,2$, it contradicts that $k$ is an integer. Therefore, $e_{4,4}=0.$ Thus, $n_4=k+2, n_2=k+4, n_3=k+1.$ By Lemma 1 and Lemma 2, $e_{2,4}=k+4, e_{3,4}=2k+2, e_{1,2}=k+4, e_{1,3}=k+1, e_{1,4}=k+2$. By the definition, $T^*\in \mathcal{G}_0$. Moreover,

\begin{flalign}
&DSO(T^*)\nonumber\\
&=\sum_{uv\in E_{1,2}}f(uv)+\sum_{uv\in E_{1,3}}f(uv)+\sum_{uv\in E_{1,4}}f(uv)+\sum_{uv\in E_{2,4}}f(uv)\nonumber\\
&~~~+\sum_{uv\in E_{3,4}}f(uv)+\sum_{uv\in E_{4,4}}f(uv)\nonumber\\
&=(\frac{2\sqrt{5}}{3}+\frac{\sqrt{10}}{4}+\frac{\sqrt{17}}{5}+\frac{10}{7})k
+\frac{\sqrt{5}}{3}+\frac{2\sqrt{17}}{5}
+\frac{\sqrt{10}}{4}+\frac{10}{7}\nonumber\\
&=(\frac{\sqrt{17}}{30}+\frac{\sqrt{5}}{9}+\frac{\sqrt{10}}{24}+\frac{5}{21})n
+\frac{10\sqrt{5}}{9}-\frac{\sqrt{17}}{15}-\frac{\sqrt{10}}{3}-\frac{40}{21}.\nonumber
\end{flalign}

Thus,$$DSO(T)\leq(\frac{\sqrt{17}}{30}+\frac{\sqrt{5}}{9}+\frac{\sqrt{10}}{24}+\frac{5}{21})n
+\frac{10\sqrt{5}}{9}-\frac{\sqrt{17}}{15}-\frac{\sqrt{10}}{3}-\frac{40}{21}.$$

\vspace{2mm}\noindent{\bf Case 2.} $n=6k+12,$ for an integer $k\geq0.$

Combining $n=6k+12$, for an integer $k\geq0,$ we have $e_{4,4}=3(n_4-k)-5.$ By Lemma 3, we have $e_{4,4}=0,1,2$. When $e_{4,4}=0,2$, it contradicts that $k$ is an integer. Therefore, $e_{4,4}=1.$ Thus, $n_4=k+2, n_2=k+4, n_3=k.$ By Lemma 1 and Lemma 2, $e_{2,4}=k+4, e_{3,4}=2k, e_{1,2}=k+4, e_{1,3}=k, e_{1,4}=k+2$. By the definition, $T^*\in \mathcal{G}_1$. Moreover,

\begin{flalign}
&DSO(T^*)\nonumber\\
&=\sum_{uv\in E_{1,2}}f(uv)+\sum_{uv\in E_{1,3}}f(uv)+\sum_{uv\in E_{1,4}}f(uv)+\sum_{uv\in E_{2,4}}f(uv)\nonumber\\
&~~~+\sum_{uv\in E_{3,4}}f(uv)+\sum_{uv\in E_{4,4}}f(uv)\nonumber\\
&=(\frac{2\sqrt{5}}{3}+\frac{\sqrt{10}}{4}+\frac{\sqrt{17}}{5}+\frac{10}{7})k
+\frac{8\sqrt{5}}{3}+\frac{2\sqrt{17}}{5}+\frac{\sqrt{2}}{2}\nonumber\\
&=(\frac{\sqrt{17}}{30}+\frac{\sqrt{5}}{9}+\frac{\sqrt{10}}{24}+\frac{5}{21})n
+\frac{4\sqrt{5}}{3}-\frac{\sqrt{10}}{2}-\frac{20}{7}+\frac{\sqrt{2}}{2}.\nonumber
\end{flalign}

Thus,$$DSO(T)\leq(\frac{\sqrt{17}}{30}+\frac{\sqrt{5}}{9}+\frac{\sqrt{10}}{24}+\frac{5}{21})n
+\frac{4\sqrt{5}}{3}-\frac{\sqrt{10}}{2}-\frac{20}{7}+\frac{\sqrt{2}}{2}.$$

\vspace{2mm}\noindent{\bf Case 3.} $n=6k+16,$ for an integer $k\geq0.$

Combining $n=6k+16,$ for an integer $k\geq0,$ we have $e_{4,4}=3(n_4-k)-7.$ By Lemma 3, we have $e_{4,4}=0,1,2$. When $e_{4,4}=0,1$, it contradicts that $k$ is an integer. Therefore, $e_{4,4}=2.$ Thus, $n_4=k+3, n_2=k+5, n_3=k.$ By Lemma 1 and Lemma 2, $e_{2,4}=k+5, e_{3,4}=2k, e_{1,2}=k+5, e_{1,3}=k, e_{1,4}=k+3$. By the definition, $T^*\in \mathcal{G}_2$. Moreover,

\begin{flalign}
&DSO(T^*)\nonumber\\
&=\sum_{uv\in E_{1,2}}f(uv)+\sum_{uv\in E_{1,3}}f(uv)+\sum_{uv\in E_{1,4}}f(uv)+\sum_{uv\in E_{2,4}}f(uv)\nonumber\\
&~~~+\sum_{uv\in E_{3,4}}f(uv)+\sum_{uv\in E_{4,4}}f(uv)\nonumber\\
&=(\frac{2\sqrt{5}}{3}+\frac{\sqrt{10}}{4}+\frac{\sqrt{17}}{5}+\frac{10}{7})k
+\frac{10\sqrt{5}}{3}+\frac{3\sqrt{17}}{5}+\sqrt{2}\nonumber\\
&=(\frac{\sqrt{17}}{30}+\frac{\sqrt{5}}{9}+\frac{\sqrt{10}}{24}+\frac{5}{21})n
+\frac{\sqrt{17}}{15}+\frac{14\sqrt{5}}{9}-\frac{2\sqrt{10}}{3}-\frac{\sqrt{80}}{21}+\sqrt{2}.\nonumber
\end{flalign}

Thus,$$DSO(T)\leq(\frac{\sqrt{17}}{30}+\frac{\sqrt{5}}{9}+\frac{\sqrt{10}}{24}+\frac{5}{21})n
+\frac{\sqrt{17}}{15}+\frac{14\sqrt{5}}{9}-\frac{2\sqrt{10}}{3}-\frac{\sqrt{80}}{21}+\sqrt{2}.$$

Summing up the above, we complete the proof.
\end{proof}

\vspace{2mm}\noindent{\bf Declarations of competing interest}

The authors declare that have no known competing financial interests or personal relationships that could have appeared to influence the work reported in this paper.

\end{document}